\documentclass[sigconf,authorversion,screen]{acmart}

\copyrightyear{2020}
\acmYear{2020}
\setcopyright{acmcopyright}
\acmConference[KDD '20]{Proceedings of the 26th ACM SIGKDD Conference on Knowledge Discovery and Data Mining}{August 23--27, 2020}{Virtual Event, CA, USA}
\acmBooktitle{Proceedings of the 26th ACM SIGKDD Conference on Knowledge Discovery and Data Mining (KDD '20), August 23--27, 2020, Virtual Event, CA, USA}
\acmPrice{15.00}
\acmDOI{10.1145/3394486.3403313}
\acmISBN{978-1-4503-7998-4/20/08}

\acmSubmissionID{aofp1990}

\AtBeginDocument{%
  \providecommand\BibTeX{{%
    \normalfont B\kern-0.5em{\scshape i\kern-0.25em b}\kern-0.8em\TeX}}}
    
\usepackage{graphicx}
\usepackage[flushleft]{threeparttable}
\usepackage{adjustbox}
\usepackage{booktabs}
\usepackage{amsmath, amsthm, mathtools}
\usepackage{color}
\usepackage{enumitem}
\usepackage{bm}
\usepackage{makecell}
\usepackage{float}
\usepackage{tikz}
\usetikzlibrary{shapes,decorations,arrows,calc,arrows.meta,fit,positioning}

\newcommand\independent{\protect\mathpalette{\protect\independenT}{\perp}}
\def\independenT#1#2{\mathrel{\rlap{$#1#2$}\mkern2mu{#1#2}}}

\newcommand\Tau{\mathrm{T}}

\newcommand{\matr}[1]{\bm{#1}} 
\newcommand{\vect}[1]{\bm{#1}}

\theoremstyle{plain}
\newtheorem{theorem}{Theorem}

\newtheorem{proposition}{Proposition}

\theoremstyle{definition}
\newtheorem{definition}{Definition}
\newtheorem{assumption}{Assumption}

\theoremstyle{remark} %
\newtheorem*{remark}{Remark} %

\usepackage{algorithm}
\usepackage{algpseudocode}
\algnewcommand\algorithmicinput{\textbf{Input}}
\algnewcommand\algorithmicoutput{\textbf{Output}}
\algnewcommand\Input{\item[\algorithmicinput]}%
\algnewcommand\Output{\item[\algorithmicoutput]}%
\begin{document}

% \fancyhead{}

%
% The "title" command has an optional parameter, allowing the author to define a "short title" to be used in page headers.
\title[Causal Meta-Mediation Analysis]{Causal Meta-Mediation Analysis: Inferring Dose-Response Function From Summary Statistics of Many Randomized Experiments}

%
% The "author" command and its associated commands are used to define the authors and their affiliations.
% Of note is the shared affiliation of the first two authors, and the "authornote" and "authornotemark" commands
% used to denote shared contribution to the research.
\author{Zenan Wang}
\affiliation{%
  \institution{UC Berkeley}
  \streetaddress{530 Evans Hall}
  \city{Berkeley}
  \state{California}
  \postcode{94720}
}
\email{zenan.wang@berkeley.edu}

\author{Xuan Yin}
\affiliation{%
  \institution{Etsy, Inc.}
  \streetaddress{117 Adams St}
  \city{Brooklyn}
  \state{New York}
  \postcode{11201}
%   \country{U.S.}
  }
\email{xyin@etsy.com}

\author{Tianbo Li}
\affiliation{%
  \institution{Etsy, Inc.}
  \streetaddress{117 Adams St}
  \city{Brooklyn}
  \state{New York}
  \postcode{11201}
%   \country{U.S.}
  }
\email{tli@etsy.com}

\author{Liangjie Hong}
\affiliation{%
  \institution{LinkedIn, Inc.}
  \streetaddress{117 Adams St}
  \city{Sunnyvale}
  \state{California}
  \postcode{94085}
%   \country{U.S.}
  }
\email{liahong@linkedin.com}%

%
% The abstract is a short summary of the work to be presented in the article.
\begin{abstract}
It is common in the internet industry to use offline-developed algorithms to power online products that contribute to the success of a business.  Offline-developed algorithms are guided by offline evaluation metrics, which are often different from online business key performance indicators (KPIs).  To maximize business KPIs, it is important to pick a north star among all available offline evaluation metrics.  By noting that online products can be measured by online evaluation metrics, the online counterparts of offline evaluation metrics, we decompose the problem into two parts.  As the offline A/B test literature works out the first part: counterfactual estimators of offline evaluation metrics that move the same way as their online counterparts, we focus on the second part: causal effects of online evaluation metrics on business KPIs.  The north star of offline evaluation metrics should be the one whose online counterpart causes the most significant lift in the business KPI.  We model the online evaluation metric as a mediator and formalize its causality with the business KPI as dose-response function (DRF).  Our novel approach, causal meta-mediation analysis, leverages summary statistics of many existing randomized experiments to identify, estimate, and test the mediator DRF.  It is easy to implement and to scale up, and has many advantages over the literature of mediation analysis and meta-analysis.  We demonstrate its effectiveness by simulation and implementation on real data.
\end{abstract}

%
% The code below is generated by the tool at http://dl.acm.org/ccs.cfm.
% Please copy and paste the code instead of the example below.
%
% \begin{CCSXML}
% <ccs2012>
% <concept>
% <concept_id>10002944.10011123.10011131</concept_id>
% <concept_desc>General and reference~Experimentation</concept_desc>
% <concept_significance>500</concept_significance>
% </concept>
% <concept>
% <concept_id>10002944.10011123.10011133</concept_id>
% <concept_desc>General and reference~Estimation</concept_desc>
% <concept_significance>500</concept_significance>
% </concept>
% <concept>
% <concept_id>10002950.10003648</concept_id>
% <concept_desc>Mathematics of computing~Probability and statistics</concept_desc>
% <concept_significance>500</concept_significance>
% </concept>
% </ccs2012>
% \end{CCSXML}

% \ccsdesc[500]{General and reference~Experimentation}
% \ccsdesc[500]{General and reference~Estimation}
% \ccsdesc[500]{Mathematics of computing~Probability and statistics}

%
% Keywords. The author(s) should pick words that accurately describe the work being
% presented. Separate the keywords with commas.

% This command processes the author and affiliation and title information and builds
% the first part of the formatted document.
\maketitle

\section{Introduction}
Nowadays it is common in the internet industry to develop algorithms that power online products using historical data.  The one that improves evaluation metrics from historical data will be tested against the one that has been in production to assess the lift in key performance indicators (KPIs) of the business in online A/B tests.  Here we refer to metrics calculated from historical data as \textit{offline} metrics and metrics calculated in online A/B tests as \textit{online} metrics.  In many cases, offline evaluation metrics are different from online business KPIs.  For instance, a ranking algorithm, which powers search pages in e-commerce platforms, typically optimizes for relevance by predicting purchase or click probabilities of items.  It could be tested offline (offline A/B tests) for rank-aware evaluation metrics, for example, normalized discounted cumulative gain ({\tt NDCG}), mean reciprocal rank ({\tt MRR}) or mean average precision ({\tt MAP}), which are calculated from the test set of historical purchase or click-through feedback of users.  Most e-commerce platforms, however, deem sitewide gross merchandise value ({\tt GMV}) as their business KPI and test for it online.  There could be various reasons not to directly optimize for business KPIs offline or use business KPIs as offline evaluation metrics, such as technical difficulty, business reputation, or user loyalty.  Nonetheless, the discrepancy between offline evaluation metrics and online business KPIs poses a challenge to product owners because it is not clear that, in order to maximize online business KPIs, which offline evaluation metric should be adopted to guide the offline development of algorithms.

The challenge essentially asks for the causal effects of increasing offline evaluation metrics on business KPIs, e.g., how business KPIs would change for a 10\% increase in an offline evaluation metric.  The offline evaluation metric in which a 10\% increase could result in the most significant lift in business KPIs should be the north star to guide algorithm development.  Algorithms developed offline power online products, and online products contribute to the success of the business (see Figure~\ref{fig: offline-online}).  By noting that online products can be measured by online evaluation metrics, the online counterparts of offline evaluation metrics, we decompose the problem into two parts.  The offline A/B test literature (see, e.g, \citet{Gilotte2018OfflineSystems}) works out the first part (the black arrow): counterfactual estimators of offline evaluation metrics to bridge the inconsistency between changes of offline and online evaluation metrics.  We focus on the second part (the red arrow): the causality between online products (assessed by online evaluation metrics) and the business (assessed by online business KPIs).  The offline evaluation metric whose online counterpart causes the most significant lift in online business KPIs should be the north star.  Hence, the question for us becomes, how business KPIs would change for a 10\% increase in an online evaluation metric.

% \begin{figure}[h]
%     \centerline{\includegraphics[scale = 0.369]{offline-online.png}}
%     \caption{The Causal Path from Algorithms to Business}
%     \label{fig: offline-online}
% \end{figure}
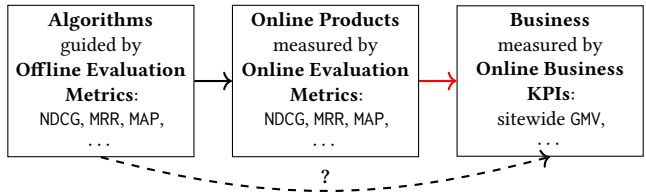
\begin{figure}[h]
    \centering
    \begin{tikzpicture}
    [% STYLES
    every node/.style={node distance=1cm, scale=0.85},
    force/.style={rectangle, draw, inner sep=3pt, text width=2.7cm, text badly centered, minimum height=1.2cm, font=\normalsize\rmfamily}
    ] 

    % Draw forces
    \node [force] (algorithm) {\textbf{Algorithms}\\guided by\\\textbf{Offline Evaluation Metrics}:\\{\tt NDCG}, {\tt MRR}, {\tt MAP},\\$\cdots$};
    
    \node [force, right=0.5cm of algorithm] (products) {\textbf{Online Products}\\measured by\\\textbf{Online Evaluation Metrics}:\\{\tt NDCG}, {\tt MRR}, {\tt MAP},\\$\cdots$};
    
    \node [force, right=0.5cm of products] (business) {\textbf{Business}\\measured by\\\textbf{Online Business KPIs}:\\sitewide {\tt GMV},\\$\cdots$};

    % Draw the links between forces
    \path[->,thick] (algorithm) edge (products);
    \path[->,thick,red] (products) edge (business);
    \path[->,thick,dashed] (algorithm.south) edge[bend right=15] node[above] {\textbf{?}} (business.south);
    \end{tikzpicture} 
\caption{The Causal Path from Algorithms to Business}
\label{fig: offline-online}
\end{figure}
Randomized controlled trials, or online A/B tests, are popular to measure the causal effects of online product change on business KPIs.  Unfortunately, they cannot answer our question directly.  In online A/B tests, in order to compare the business KPIs caused by different values of an online evaluation metric, we need to fix the metric at its different values for treatment and control groups.  Take the ranking algorithm as an example.  If we could fix online {\tt NDCG} of the search page at 0.22 and 0.2 for treatment and control groups respectively, then we would know how sitewide {\tt GMV} would change for a 10\% increase in online {\tt NDCG} at 0.2.  However, this experimental design is impossible, because most online evaluation metrics depend on users' feedback and thus cannot be directly controlled.

We address the question by developing a novel approach of causal inference.  We model the causality between online evaluation metrics and business KPIs by \textit{dose-response function} ({\tt DRF}) in potential outcome framework~\citep{Imbens2000TheFunctions,Imbens2004TheTreatments}.  {\tt DRF} originates from medicine and describes the magnitude of the response of an organism given different doses of a stimulus.  Here we use it to depict the value of a business KPI given different values of an online evaluation metric.  Different from doses of stimuli, values of online evaluation metrics cannot be directly manipulated.  However, they could differ between treatment and control groups in experiments of treatments other than algorithms---user interface/user experience (UI/UX) design, marketing, etc.  This could be due to the ``fat hand''~\citep{Peysakhovich,Yin2019TheAnalysis} nature of online A/B tests that a single intervention can change many causal variables at once.  A change of the tested feature, which is not algorithm, could induce users to change their engagement with algorithm-powered online products, so that values of online evaluation metrics would change.  For instance, in an experiment of UI design, users might change their search behaviors because of the new UI design, so that values of online {\tt NDCG}, which depends on search interaction, would change, even though ranking algorithm does not change.  The evidence suggests that online evaluation metrics could be \textit{mediators} that (partially) transmit causal effects of treatments on business KPIs in experiments where treatments are not necessarily algorithm-related.  Hence, we formalize the problem as the identification, estimation, and test of mediator {\tt DRF}.

In mediation analysis literature, there are two popular identification techniques: \textit{sequential ignorability} ({\tt SI}) and \textit{instrumental variable} ({\tt IV}).  {\tt SI} assumes each potential mediator is independent of all potential outcomes conditional on the assigned treatment, whereas {\tt IV} permits dependence between unknown factors and mediators but forbids the existence of direct effects of the treatment.  Rather than making these stringent assumptions, we leverage trial characteristics to explain \textit{average direct effect} ({\tt ADE}) in each experiment so that we can tease it out from \textit{average treatment effect} ({\tt ATE}) to identify the causal mediation. The utilization of trial characteristics means we have to use data from many trials because we need variations in trial characteristics. Hence, we develop our framework as a meta-analysis and propose an algorithm that only uses summarized results from many existing experiments and gain the advantage of easy implementation to scale.

Most meta-analyses rely on summarized results from different studies with different raw data sources.  Therefore, it is almost impossible to learn more beyond the distribution of {\tt ATE}s. Fortunately, the internet industry produces plentiful randomized trials with consistently defined metrics, and thus presents an opportunity for performing a more complicated meta-analysis.  Literature is lacking in this area while we create the framework of \textit{causal meta-mediation analysis} ({\tt CMMA}) to fill in the gap.

Another prominent strength of our approach in real application is, for a new product that has been shipped online but has few A/B tests, it is plausible to explore the causality between its online metrics and business KPIs from many A/B tests of other products.  The values of online metrics of the new product can differ between treatment and control groups in experiments of other products ("fat hand"~\citep{Peysakhovich,Yin2019TheAnalysis}), which makes it possible to solve for mediator {\tt DRF} of the new product without its own A/B tests.  
% The application depends on the technical conditions of our approach and many other factors, which is worth to investigate in the future.

Note that, our approach can be applied to any evaluation metric that is defined at experimental-unit level, like metrics discussed in offline A/B test literature.  The experimental unit means the unit for randomization in online A/B tests.  For example, in search page experiments, the experimental unit is typically the user.  Also, the evaluation metric can be any combination of existing experimental-unit-level metrics.

To summarize, our contributions in this paper include:
\begin{enumerate}[nolistsep]
    \item This is the first study that offers a framework to choose the north star among all available offline evaluation metrics for algorithm development to maximize business KPIs when offline evaluation metrics and business KPIs are different.  We decompose the problem into two parts.  Since the offline A/B test literature works out the first part: counterfactual estimators of offline evaluation metrics to bridge the inconsistency between changes of offline and online metrics, we work out the second part: inferring causal effects of online evaluation metrics on business KPIs.  The offline evaluation metric whose online counterpart causes the most significant lift in business KPIs should be the north star.  We show the implementation of our framework on data from Etsy.com.
    \item Our novel approach {\tt CMMA} combines mediation analysis and meta-analysis to identify, estimate, and test mediator {\tt DRF}.  It relaxes standard {\tt SI} assumption and overcomes the limitation of {\tt IV}, both of which are popular in causal mediation literature.  It extends meta-analysis to solve causal mediation while the meta-analysis literature only learns the distribution of {\tt ATE}.  We demonstrate its effectiveness by simulation and show its performance is superior to other methods.
    \item Our novel approach {\tt CMMA} uses only trial-level summary statistics (i.e., meta-data) of many existing trials, which makes it easy to implement and to scale up.  It can be applied to all experimental-unit-level evaluation metrics or any combination of them.  Because it solves for causality problem of a product by leveraging trials of all products, it could be particularly useful in real applications for a new product that has been shipped online but has few A/B tests.
\end{enumerate}
% The remainder is organized as follows.  In Section~\ref{sec:literature_review} we reviews the literature.  In Section~\ref{sec:conceptual_framework} we present conceptual framework. Our core result, the identification of mediator {\tt DRF} is presented in Section~\ref{sec:identification}. Then we discuss the estimation and hypothesis testing in Section~\ref{sec:estimation}. We conduct simulation for model validation and comparison in Section~\ref{sec:simulation}, and apply our approach to real data in Section~\ref{sec:application}.
%
\section{Literature Review}\label{sec:literature_review}
We draw on two strands of literature: mediation analysis and meta-analysis.  We briefly discuss them in turn.
\subsection{Mediation Analysis}
Our framework expands on causal mediation analysis.  Mediation analysis is actively conducted in various disciplines, such as psychology~\citep{MacKinnon2006MediationAnalysis,Rucker2011MediationRecommendations}, political science~\citep{green2010,Imai2010}, economics~\citep{Heckman2015EconometricInputs}, and computer science~\citep{Pearl2001DirectEffects}.  The recent application in the internet industry reveals the performance of recommendation system could be cannibalized by search in e-commerce website~\citep{Yin2019TheAnalysis}.  Mediation analysis originates from the seminal paper of~\citet{Baron1986}, where they proposed a parametric estimator based on the linear structural equation model ({\tt LSEM}).  {\tt LSEM}, by far, is still widely used by applied researchers because of its simplicity.  Since then, \citet{Robins1992IdentifiabilityEffects} and \citet{Pearl2001DirectEffects} and other causal inference researchers have formalized the definition of causal mediation and pinpointed assumptions for its identification~\citep{Robins2003,Robins2010AlternativeEffects,Pearl2014} in various complicated scenarios.  The progress features extensive usage of structural equation models and causal diagrams (e.g., {\tt NPSEM-IE} of \citet{Pearl2001DirectEffects} and {\tt FRCISTG} of \citet{Robins2003}). 

As researchers extend the potential outcome framework of \citet{Rubin2003BasicStudies} to causal mediation, alternative identification, and more general estimation strategies have been developed.  \citet{Imai2010} achieved the non-parametric identification and estimation of mediation effects of a single mediator under the assumption of {\tt SI}. After analyzing other well-known models such as {\tt LSEM}~\citep{Baron1986} and {\tt FRCISTG}~\citep{Robins2003}, they concluded that assumptions of most models can be either boiled down to or replaced by {\tt SI}.  However, {\tt SI} is stringent, which ignites many in-depth discussions around it (see, e.g., the discussion between \citet{Pearl2014, Pearl2014ReplyAnalysis} and \citet{Imai2014CommentAnalysis.}). 
% \citet{Imai2013IdentificationExperiments} worked out the identification for multiple causally-dependent mediators 
% when there is no unmeasured post-treatment confounder.  
% {\tt SI}, which is proposed by \citet{Imai2010}, is perhaps the most popular identification assumption in mediation analysis now.
% \citet{Imai2010} analyzed other well-known models such as {\tt FRCISTG} of \citet{Robins2003} and concluded that assumptions of most models can be either boiled down to or replaced by {\tt SI}.

Another popular identification strategy of causal mediation is {\tt IV}, which is a signature technique in economics~\citep{Angrist1996IdentificationVariables,Angrist2001InstrumentalExperiments}.  \citet{Sobel2008IdentificationVariables} used treatment as {\tt IV} to identify mediation effects without {\tt SI}. However, as \citet{Imai2010} pointed out, {\tt IV} assumptions may be undesirable because they require all causal effects of the treatment pass through the mediator (i.e., complete mediation~\citep{Baron1986}). \citet{Small2012MediationVariables} proposed a new method to construct {\tt IV} that allows direct effects of the treatment (i.e., partial mediation~\citep{Baron1986}) but assumes that {\tt ADE} of the treatment is the same for different segments of the population.

\subsection{Meta-Analysis}
% \citep{Kolesar2015IdentificationInstruments}
Our method only uses summary statistics of many past experiments.  Analyzing summarized results from many experiments is termed as meta-analysis and is common in analytical practice~\citep{Cooper2009TheMeta-analysis,Stanley2012Meta-regressionBusiness}.  In the literature, meta-analysis is used for mitigating the problem of external validity in a single experiment and learning knowledge that was hard to recover when analyzing data in isolation, such as heterogeneous treatment effects (see, e.g., \citet{Higgins2002QuantifyingMeta-analysis,Browne2017}).  Besides, a significant advantage of meta-analysis is easy to scale, because it only takes summarized results from many different experiments.

\citet{Peysakhovich} took one step toward the direction of performing mediation analysis using data from many experiments. They used treatment assignments as {\tt IV}s to identify causal mediation, which is similar to \citet{Sobel2008IdentificationVariables}, but lacks the justification why more than one experiment is needed and failed to address limitations of {\tt IV} that we discussed above. Our framework shows that having access to many experiments enables identifying causal mediation without {\tt SI} and overcoming the limitation of {\tt IV}, both of which are hard to achieve with only one experiment.

\section{Conceptual Framework}\label{sec:conceptual_framework}
We follow the literature of potential outcomes~\citep{Rubin2003BasicStudies,Imai2010,Small2012MediationVariables} to set up our framework.

\begin{figure}[h]
    \centering
    
% \resizebox{1.5in}{1.4in}{%
\begin{tikzpicture}
    [
    % Tikz settings optimized for causal graphs.
    -Latex,auto,node distance =1 cm and 1 cm,semithick,
    state/.style ={ellipse, draw, minimum width = 0.7 cm},
    point/.style = {circle, draw, inner sep=0.04cm,fill,node contents={}},
    bidirected/.style={Latex-Latex,dashed},
    el/.style = {inner sep=2pt, align=left, sloped}
    ]
    % x node set with absolute coordinates
    \node[state] (m) at (0,0) {$M$};

    % y node set relative to x.
    % Locations can be:
    % right,left,above,below,
    % above left,below right, etc
    \node[state] (y) [right =of m] {$Y$};
    \node[state] (t1) at ([shift={(m)}] 120:2) {$T^{s_1}$};
    \node[state] (t2) at ([shift={(m)}] 150:2) {$T^{s_2}$};
    \node[state] (t3) at ([shift={(m)}] 180:2) {$T^{s_3}$};
    \node[state] (tN) at ([shift={(m)}] 240:2) {$T^{s_K}$};
    \node[state] (w) at ([shift={(m)}] -60:2){$U$};
    \node (dot1) at ([shift={(m)}] 210:2){\Large.};
    \node (dot2) at ([shift={(m)}] 220:2){\Large.};
    \node (dot3) at ([shift={(m)}] 200:2){\Large .};

    % Directed edge
    \path[red] (m) edge (y);
    \path (t1) edge (m);
    \path (t2) edge (m);
    \path (t3) edge (m);
    \path (tN) edge (m);
    \path[dashed] (w) edge (m);
    \path[dashed] (w) edge (y);
    \path (t1) edge[bend left=30] (y);
    \path (t2) edge[bend left=30] (y);
    \path (t3) edge[bend left=30] (y);
    \path (tN) edge[bend right=20] (y);
    \end{tikzpicture}
% }
\caption{Directed Acyclic Graph of Conceptual Framework}
\label{fig:framework}
\end{figure}
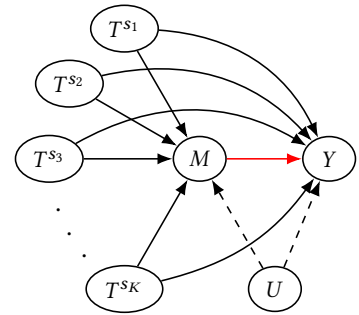

As illustrated in Figure~\ref{fig:framework}, we suppose there are many experiments with different treatments, and there is a mediator $M$ that can be affected by any treatment $T^s$ and will in turn influence outcome $Y$. Each treatment may also affect $Y$ directly. But we are particularly interested in recovering their shared casual channel, the link between $M$ and $Y$, marked in red in the figure.  In the ranking algorithm example, $M$ is online evaluation metric of search page, $Y$ is online business KPI.  A challenge to identify the red link emerges if a confounder $U$ exists.  $U$ could be (user engagement of) any other web-page/module or user preference in the ranking algorithm example. In the literature, there are two approaches to solve this challenge, {\tt SI} and {\tt IV}. {\tt SI} requires that $U$ is observed and measured, and there are no other unmeasured/unobserved confounders after controlling for $U$. The standard {\tt IV} approach allows unmeasured/unobserved $U$, but assumes no direct links between $T^{s}$s and $Y$.  An {\tt IV} method proposed by \citet{Small2012MediationVariables} relaxes requirements of standard approach, and assumes all $T^{s}$s share a single direct link to $Y$. Our method allows unmeasured/unobserved $U$ and direct links from each $T^{s}$ to $Y$.

Suppose there are $K$ randomized trials in total. For each trial, there exists one treatment group and one control group\footnote{Experiments with multiple treatment arms can be considered as multiple trials with one treatment group and one control group.}. To simplify the discussion, here we assume experimental units are first randomly assigned to different trials, and then randomly assigned into the treatment or control group of that trial.  However, our approach {\tt CMMA} allows the same unit to participate in multiple trials.  We will go back to this point in Section~\ref{sec:algorithm}.

We consider the following model for potential outcomes.
\begin{definition}\label{def1}
    \begin{align}
        M_i^{(t,\vect{s})} &= \vect{\tau_i}^{\top}\vect{s}\times t + \vect{\phi}^{\top}\vect{s} +  M_{i}^{*} \label{def1: eq1}\\
        Y_i^{(t,\vect{s},m)}&= \mu_i(m) + \vect{\gamma_i}^{\top}\vect{s}\times t+ \vect{\theta}^{\top}\vect{s} + Y_{i}^{*},\label{def1: eq2}
   \end{align}
   where the ($\vect{\tau_i}$, $\vect{\gamma_i}$, $M_{i}^{*}$, $Y_{i}^{*}$) are i.i.d random vectors,
   for all $t=0 \text{ or } 1$, $\vect{s}\in\mathcal{S}$ and $m \in \mathcal{M}$.
\end{definition}
We use a one-hot vector $\vect{s} = (s_1,\cdots,s_K)$ to encode the trial assignment where $(0,\cdots, s_k=1, \cdots,0)$ indicates the assignment to trial $k$, and a binary variable $t \in \{0, 1\}$ to encode treatment assignment, with $t=1$ if assigned to the treatment group.  For experimental unit $i$ with trial assignment $\vect{s}$ and treatment assignment $t$, the random variable $M_i^{(t,\vect{s})}$ represents the potential mediator, and the random variable $Y_i^{(t,\vect{s}, m)}$ represents the potential outcome that would be observed if $i$ were to receive or exhibit level $m$ of the mediator through some hypothetical mechanism. 

We are interested in mediator {\tt DRF}, which represents the value of $Y_i^{(t,\vect{s}, m)}$ given different values ($m$) of $M_i^{(t,\vect{s})}$ for the same $t$ and $\vect{s}$. This effectively is the red arrow in the Figure~\ref{fig:framework} for individuals.  Let $\mu_i(m)$ be the mediator {\tt DRF} for individual $i$. Our goal is to estimate average mediator {\tt DRF}: $\mathbb{E}[\mu_i(m)]$, with which we can compute the percentage change of $\mathbb{E}[\mu_i(m)]$ for a 10\% increase in $m$ \textit{ceteris paribus}. The expectation here is taken over the population of $i$ and so are other expectations in this paper if not specified. Here we consider polynomial mediator {\tt DRF}: $\mu_i(m) = \sum_{p=1}^{P}\beta_{i,p} m^p$, which can capture the nonlinearity of the causality.  Estimating $\mathbb{E}[\mu_i(m)]$ means to estimate $\beta_p = \mathbb{E}[\beta_{i,p}]$ for $p = 1, 2,\cdots, P$.  Vectors $\vect{\tau_i}$ and $\vect{\gamma_i}$ are in $\mathbb{R}^K$. Each element of the vector, $\tau_{i,k}$ or $\gamma_{i,k}$, represents direct effect of the treatment on mediator or outcome in trial $k$, and is assumed not to depend on $m$. 

Vectors $\vect{\phi}$ and $\vect{\theta}$ are also in $\mathbb{R}^K$, representing trial fixed effects. Random variables $M_{i}^{*}$ and $Y_{i}^{*}$ represent idiosyncratic individual characteristics in the values of potential mediator and potential outcome. Assume $\mathbb{E}[M_{i}^{*}] = \mathbb{E}[Y_{i}^{*}] = 0$ \footnote{This assumption can always be satisfied by reparameterizing $M_{i}^{*}, Y_{i}^{*}$ and $\vect{\phi}, \vect{\theta}$.}. Let $\mathcal{M}$ be the support of the distribution of mediator, and $\mathcal{S}$ be the set containing all possible trial assignments.

We only observe realized data $\{M_i, Y_i, T_i, \vect{S}_i\}$. The observed mediator $M_i := M_i^{(T_i,\vect{S}_i)}$, and observed outcome $Y_i := Y_i^{(T_i,\vect{S}_i,M_i)}$.

The specification has two implications. First, it implies that being in a particular trial will not affect mediator {\tt DRF}. If mediator {\tt DRF} is not trial independent, then having many trials only adds noises rather than provides additional information, and thus defeats the purpose of conducting a meta-analysis.
\begin{remark}[Trial-Irrelevant Mediator {\tt DRF}]\label{assump:trial_independent}
\begin{align}
    Y_i^{(t,\vect{s},m')}-Y_i^{(t,\vect{s},m)} = \mu_i(m')-\mu_i(m), 
\end{align}
for all $t=0 \text{ or }1$, $m', m\in\mathcal{M}$ and $\vect{s}\in\mathcal{S}$, where $\mu_i(m')-\mu_i(m)$ does not depend on $s$.
\end{remark}
Second, the specification implies that there are no interaction effects between treatment and mediator on the outcome.  It means the individual direct effect in each trial is irrelevant of the value of mediator. It is common in the literature of causal mediation (see, e.g., {\tt NPSEM-IE} of \citet{Pearl2001DirectEffects} and {\tt FRCISTG} of \citet{Robins2003}). 
\begin{remark}[No-Interaction \citet{Pearl2001DirectEffects,Robins2003}]\label{assump:no-interaction}
\begin{align}
    Y_i^{(1,\vect{s},m)}-Y_i^{(0,\vect{s},m)} =\vect{\gamma_i}^{\top}\vect{s}, 
\end{align}
for all $m\in\mathcal{M}$ and $\vect{s}\in\mathcal{S}$, where $\vect{\gamma_i}$ does not depend on $m$.
\end{remark}
\begin{figure}
    \includegraphics[width = 0.25\textwidth]{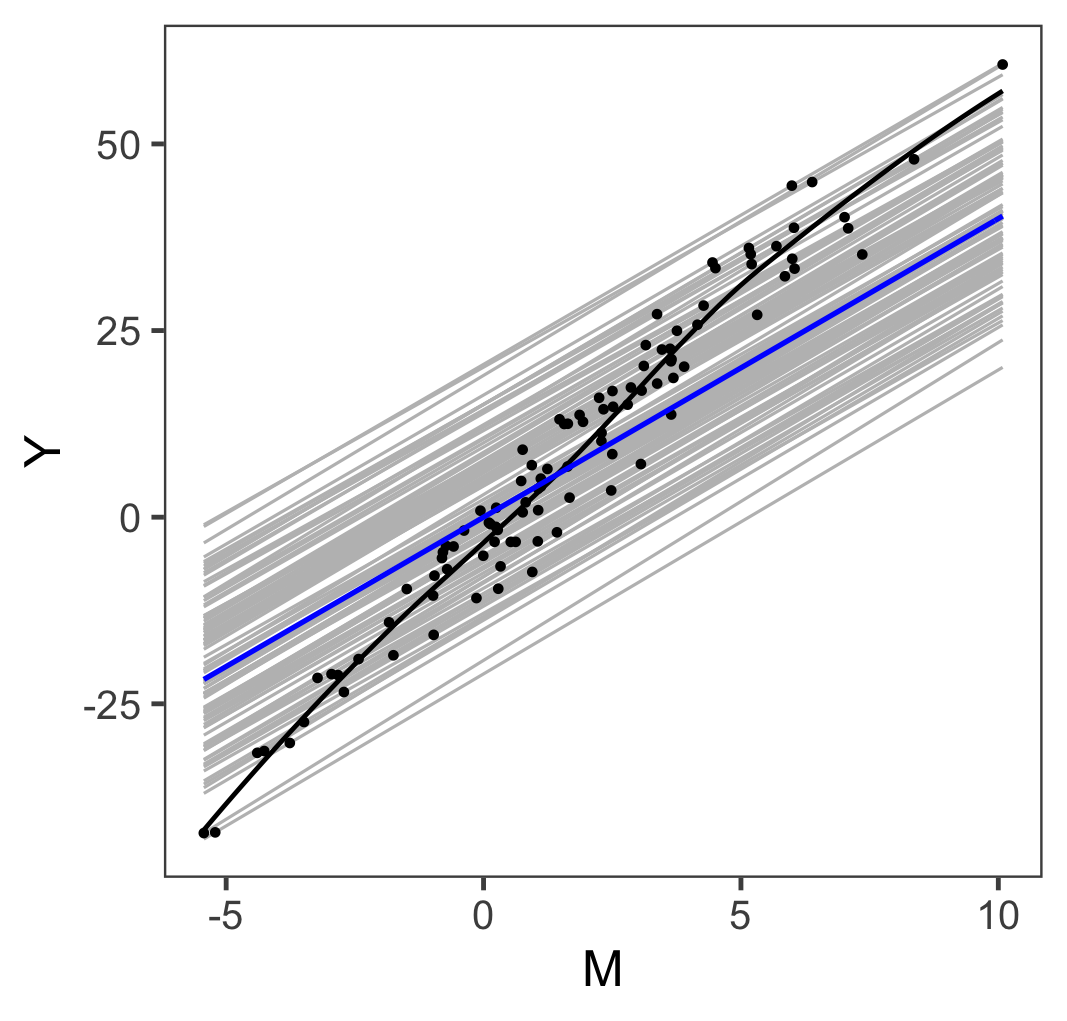}
    \caption{Mediator {\tt DRF} and Observed Data in Control Group}
    \label{fig:po}
\end{figure}

% In this paper, we are interested in $\beta = \mathbb{E}[\beta_i]$, which leads to average mediator {\tt DRF}, the average effect of a one unit increase in the
% mediator on the outcome \textit{ceteris paribus}. The average here is taken over the population of interest and so is other expectations in this paper if not specified. 
When mediator $M_i$ is not independent of $Y_i$ (which is the case here), it is not trivial to recover $\mathbb{E}[\mu_i(m)]$ through observed data. 
% Using only data from the control group of a single trial, one could estimate conditional expectation $$\mathbb{E}[Y_i^{(0,\vect{s}_1,m)}|M_i=m]-\mathbb{E}[Y_i^{(0,\vect{s}_1,m')}| M_i'=m'].$$ But this produces $\mathbb{E}[\beta_i|M_i=m]-\mathbb{E}[\beta_i|M_i'=m']$, not equal to $\mathbb{E}[\beta_i]$.
Figure~\ref{fig:po} illustrates the challenge using simulated data. The grey lines in Figure~\ref{fig:po} are simulated individual mediator {\tt DRF}s following Definition~\ref{def1}, which represents the true causality between $M_i$s and $Y_i$s. The blue line is the average mediator {\tt DRF}. After randomly assign individuals to a trial, we can compute their observed mediator values $M_i$ and outcome values $Y_i$ when in the control group, which are depicted by black scattered points. The black line shows the result from fitting the observed points by a widely-used non-parametric machine learning algorithm, locally estimated scatterplot smoothing (LOESS). Although the black line fits the data almost perfectly, it significantly deviates from the true underlying causality: the average mediator {\tt DRF}.

\section{Identification}\label{sec:identification}
% \textcolor{red}{In this section, we propose a method utilizing the unique feature of having many randomized trials to identify mediator {\tt DRF} in the presence of both unobserved confounders and direct treatment effects, which break standard {\tt SI} and {\tt IV} assumptions. We prove that, with a relaxed {\tt SI}, a trial-level conditional independence assumption, the common {\tt LSEM}, relevance and random assignment assumption, mediator {\tt DRF} can be identified through a linear regression on meta-data. The meta-data refers to estimated treatment effects on various variables for each trial and is likely already calculated in standard A/B tests. }

\subsection{Random Assignment}
We first formally define the assumption that trial assignment and treatment assignment are random. 
\begin{assumption}[Random Assignment]\label{assump:random-ts}
\begin{align}
\left\lbrace Y_i^{(t',\vect{s}, m)}, M_i^{(t,\vect{s})}\right\rbrace &\independent \vect{S_i} \label{eq:random-S}& \\
\left\lbrace Y_i^{(t',\vect{s}, m)}, M_i^{(t,\vect{s})}\right\rbrace &\independent T_i | \vect{S}_i=\vect{s} \label{eq:random-T}&
\end{align}
for $t,t' = 0 \text{ or }1$, all $m \in \mathcal{M}$ and $\vect{s} \in \mathcal{S}$ and it is also assumed that $0<\mathbb{P}(T_i=t|\vect{S}_i=\vect{s})<1 \text{ and }0<\mathbb{P}(M_i(t)=m \vert  T_i=t, \vect{S}_i=\vect{s})<1$ for $t = 0 \text{ or }1$, all $m \in \mathcal{M}$ and $\vect{s} \in \mathcal{S}$. 
\end{assumption}

Equation~\ref{eq:random-T} can be guaranteed by random assignment of treatments in online A/B tests.  Equation~\ref{eq:random-S} means that individual's potential outcomes and potential mediators are independent of the trial assignment. In practice, users are randomly selected into A/B tests, thus this assumption is trivially satisfied \footnote{This is true regardless of whether the same unit can participate in multiple trials.}.
\subsection{Relaxed Sequential Ignorability}
\begin{assumption}[Relaxed Sequential Ignorability]\label{assump:relaxed-SI}
\begin{align}
 Y_i^{(t,\vect{s}, m')}-Y_i^{(t,\vect{s}, m)} \independent 
 M_i| T_i=t, \vect{S}_i=\vect{s} \label{eq:relaxed-SI}&
\end{align}
for all $m, m'\in\mathcal{M}$, all $\vect{s}\in\mathcal{S}$, $t=0,1$.
\end{assumption}
With Definition~\ref{def1} and assumption that $\mu_i(m)$ is a polynomial function, this assumption is equivalent to $$\beta_{i,1},\cdots, \beta_{i,P} \independent
M_i | T_i=t, \vect{S}_i=\vect{s}$$ for all $\vect{s}\in\mathcal{S}$, $t=0,1$.
This assumption says that the effect of changing $m$ on the outcome of $i$ is independent of the idiosyncratic individual unobservable that affects $M_i$. A similar assumption is proposed by \citet[(IV-A3)]{Small2012MediationVariables}. It means that the underlying causality between online product and the business is invariant even we observe some users produce higher online metrics than others for unknown reasons (i.e., the idiosyncratic unobservable).  It is weaker than the {\tt SI}. Whenever {\tt SI} is satisfied, this assumption is naturally satisfied.\footnote{If  $Y_i^{(t,\vect{s}, m)}\independent M_i| T_i=t, S_i=s$ for all $m\in\mathcal{M}$ then naturally $Y_i^{(t,\vect{s}, m')}-Y_i^{(t,\vect{s}, m)}$ is also independent of  $M_i$ given $T_i=t, S_i=s$.} But the inverse is not true. It is possible to break {\tt SI} and still fulfill this assumption. For example, when 
the mediator {\tt DRF} is the same for all individuals and potential outcomes only depends on unobserved confounders, {\tt SI} could be violated, while this assumption is still satisfied.

\subsection{Trial-level Conditional Independence}
In order to allow the presence of many direct treatment effects, we put some structure on the direct treatment effects. Let $H_{k}$ represents a vector of trial characteristics for trial $k$.
% Let $\gamma_{k}$ and $\tau_{k}$ represent $\mathbb{E}[\gamma_{i,k}]$ and $\mathbb{E}[\tau_{i,k}]$ respectively, where the expectations are taken over population of individuals, thus can still be treated as an random variable on the trial level.
\begin{assumption}[Trial-Level Conditional Independence]\label{assump:trial_independent}
    We assume $\gamma_{i,k}|H_k$'s are independently and identically distributed with\\
    $\mathbb{E}[\gamma_{i,k}|H_{k}] = H_{k}^{\top}\vect{\pi}$  and $\gamma_{i,k} \independent  \tau_{i,k}\;| H_{k}$ for all $i$.
\end{assumption}
% \begin{assumption}[Trial-Level Conditional Ignorability]\label{assump:trial_independent}
%     We assume $\gamma_{k}|H_k$'s are independent and identically distributed, with $\mathbb{E}[\gamma_{k}|H_{k}] = H_{k}^{\top}\vect{\pi}$ and $\gamma_{k} \independent  \tau_{k}\;| H_{k}$.
% \end{assumption}
This assumption allows correlation between individual direct treatment effects on outcome and on mediator in a trial, while assuming such correlation disappears once conditioned on the characteristics of the trial. In the example of ranking algorithm, we may believe that experiments that test new algorithms generally have high impacts on both online {\tt NDCG} and {\tt GMV} whereas experiments that test new UI designs generally have only modest impacts on both metrics. But, within the same type of experiments, how much a treatment affects online {\tt NDCG} does not correlate with its effect on {\tt GMV}. This assumption is unverifiable. However, it is weaker than standard {\tt IV} assumption in the literature.  If there were no direct effects ({\tt IV} assumption), then this assumption is trivially satisfied.  
% In the literature, \citet{Kolesar2015IdentificationInstruments} proposes a stronger zero correlation assumption that assumes independence between $\gamma_{i,k}$ and $\tau_{i,k}$ unconditionally. We relax their assumption.
% Furthermore, compared with the individual level covariates, trial characteristics are much easier to collect and store.

We can stack vectors of trial characteristics of all the $K$ trials into a matrix $\vect{H} = \begin{bmatrix}
    H_1 & H_2 & \hdots& H_K
\end{bmatrix}^{\top}$. Assumption~\ref{assump:trial_independent} implies
\begin{align}
% \mathbb{E}[\vect{\gamma} - \vect{H}\vect{\pi}] = \mathbb{E}[\mathbb{E}[\vect{\gamma}|\vect{H}] - \vect{H}\vect{\pi}] =\vect{0}.\label{eq:projection}
% \mathbb{E}[\vect{\gamma}_i] = \mathbb{E}[\mathbb{E}[\vect{\gamma}_i|\vect{H}]] = \vect{H}\vect{\pi}.\label{eq:projection}
\mathbb{E}[\vect{\gamma}_i|\vect{H}] = \begin{bmatrix}
    \mathbb{E}[\gamma_{i,1}|H_1]& \mathbb{E}[\gamma_{i,2}|H_2] & \hdots& \mathbb{E}[\gamma_{i,K}|H_K]
\end{bmatrix}^{\top} = \vect{H}\vect{\pi}\label{eq:projection}
\end{align}
for all $i$.
\subsection{Relevance Condition}
\begin{assumption}[Relevance Condition]\label{assump:first-stage}
    \begin{align}
    Var(\tau_{k}|H_k)>0;
    \end{align}
\end{assumption}
This assumption means, for the same type of experiments, direct treatment effect on $M$ ($\tau_{k}$) varies between experiments.  It implies that treatment in each trial is still helpful for predicting $M$ after conditioning on trial-level covariates.  This assumption is similar to the standard rank condition of IV identification (see \citet[Chapter 5]{Wooldridge2010EconometricData}).
% It guarantees that the instruments can introduce enough exogenous variation to $M_i$ after conditioned on trial-level covariates.
Because $\tau_k$ can be calculated easily by summary statistics, this assumption can be empirically verified. In practice, since we can decide which trials to be included into the analysis, we can make sure this assumption always holds.

\subsection{Identification of Mediator {\tt DRF}}
% The following theorem presents our main identification results, showing that parameters of a polynomial mediator {\tt DRF} can be identified through a regression on trial-level meta-data. 
To simplify the proof, let's assume $\mu_i(m) = \beta_i m$. The same proof works for the more general polynomial {\tt DRF}.
Based on Definition~\ref{def1}, random variables of observed mediator $M_i$ and observed outcome $Y_i$ can be written as
\begin{align}
    M_i&= (\vect{S}_{i}T_{i})^{\top}\vect{\tau} + \vect{S}_{i}^{\top}\vect{\phi} + \eta_i, \label{eq: obs_m}\\
    Y_i&= M_{i}\beta  + \vect{S}_{i}^{\top}\vect{\theta} + (\vect{S}_iT_{i})^{\top}\vect{\gamma}+ \epsilon_i, \label{eq: obs_y}
\shortintertext{where}
    &\eta_i =(\vect{S_i}T_{i})^{\top}(\vect{\tau_i}-\vect{\tau}) + M_{i}^{*},\nonumber \\
    &\epsilon_i = (\vect{S_i}T_{i})^{\top}(\vect{\gamma}_i-\vect{\gamma}) + M_{i}(\beta_i-\beta) + Y_{i}^{*},\nonumber \\
    &\beta = \mathbb{E}[\beta_i], \vect{\tau}=\mathbb{E}[\vect{\tau}_i], \text{and } \vect{\gamma} = \mathbb{E}[\vect{\gamma}_i]. \nonumber
\end{align}

By plugging Equation~\ref{eq: obs_m} into Equation~\ref{eq: obs_y} for $M_i$, we obtain
\begin{align}
    Y_i = (\vect{S_i}T_{i})^{\top}\vect{\tau}\beta  + \vect{S}_{i}^{\top}(\vect{\theta}+\vect{\phi}\beta) + (\vect{S}_iT_{i})^{\top}\vect{H}\vect{\pi} + \epsilon_i' \label{eq: iv_obs_y1}
\end{align}
where $\epsilon_i' = (\vect{S}_iT_{i})^{\top}(\vect{\gamma}_i - \vect{H}\vect{\pi})+ \eta_i\beta  + M_{i}(\beta_i-\beta) + Y_{i}^{*} $.

With all the specifications and assumptions, we are ready to present the important result about the identification.
% The following theorem shows that $\beta$ can be identified from equation \ref{eq: iv_obs_y1} under aforementioned assumptions.
\begin{theorem}\label{thm:identification}
    Consider the model specified in Definition~\ref{def1},  Under Assumption~\ref{assump:random-ts}~-~\ref{assump:first-stage}, the average mediator effect $\beta = \mathbb{E}[\beta_{i}]$ can be identified.  A consistent estimator of $\beta$ can be derived through a procedure of two-stage least squares (2SLS):
    \begin{enumerate}
        \item Following Equation~\ref{eq: obs_m}, regress $M_i$ on $\vect{S}_{i}T_i$ and $\vect{S}_{i}$ via least squares to obtain a consistent estimator $\widehat{\vect{\tau}}$;
        \item Plug the consistent estimator $\widehat{\vect{\tau}}$ into Equation~\ref{eq: iv_obs_y1}, and regress $Y_i$ on $(\vect{S_i}T_{i})^{\top}\widehat{\vect{\tau}}$, $\vect{S}_{i}$, and $(\vect{S_i}T_{i})^{\top}\vect{H}$ via least squares to obtain the consistent estimator $\widehat{\beta}$.
    \end{enumerate}
\end{theorem}
The proof is in Appendix~\ref{appendix:proof-thm}. Since $\vect{\tau}$ can be identified from \ref{eq: obs_m}, we can use $\vect{\widehat{\tau}}$ in place of $\vect{\tau}$. The proof shows that under Assumption~\ref{assump:random-ts}~-~\ref{assump:first-stage}, the covariances between $\epsilon_i'$ and all the covariates, $\vect{S}_{i}T_i\widehat{\tau}$, $\vect{S}_{i}$ and $\vect{S}_{i}T_{i}^{\top}\vect{H}$ in Equation~\ref{eq: iv_obs_y1} are zeros. Therefore, structural parameters $\beta$ can be identified (see, e.g., \citet[Chapter 4]{Wooldridge2010EconometricData} for more details of coefficient identification in linear regression).

The estimator $\widehat{\beta}$ from 2SLS of Theorem~\ref{thm:identification} is equivalent to an {\tt IV}-2SLS estimator (see, e.g., \citet[Chapter 5]{Wooldridge2010EconometricData} for more details of {\tt IV}-2SLS).  To see it, let's rewrite Equation~\ref{eq: obs_y} as
\begin{align}
    Y_i = M_i\beta  + \vect{S}_{i}^{\top}(\vect{\theta}+\vect{\phi}\beta) + (\vect{S}_iT_{i})^{\top}\vect{H}\vect{\pi} + \epsilon_i'' \label{eq: iv_obs_y2}
\end{align}
where $\epsilon_i'' = (\vect{S}_iT_{i})^{\top}(\vect{\gamma}_i-\vect{H}\vect{\pi}) + M_{i}(\beta_i-\beta) + Y_{i}^{*}$.  We could get the same estimator of $\beta$ as in Theorem~\ref{thm:identification} through applying {\tt IV}-2SLS on Equation~\ref{eq: iv_obs_y2} and using $\vect{S}_{i}T_{i}$ as instruments for $M_i$. 

% The first step in the procedure then becomes regressing $\overline{\vect{M}_i}$ on $\vect{S}_{i}T_i$ and $\vect{S}_{i}$. 
% \textcolor{red}{I AM NOT CLEAR WHERE TO PLACE WORDS IN BELOW\\
% If $Cov(M_i^{*},Y_i^{*}) \neq 0$, then $Cov(M_i,\epsilon_i) \neq 0$, meaning we are unable to recover model parameter $\beta$ in Equation \ref{eq: obs_y} using OLS. If $\vect{\gamma} \neq \vect{0}$, the presence of $(\vect{S}_iT_{i})^{\top}\vect{\gamma}$ means that $\vect{S}_iT_{i}$ can not be directly used as an IV because its exclusion restriction is violated.
% }

%
\section{Estimation and Hypothesis Testing}\label{sec:estimation}
\subsection{Estimation}\label{sec:algorithm}

Theorem~\ref{thm:identification} implies that average mediator effect can be estimated by running two regressions with pooled data.  Such estimation could be very costly when the sample size in each trial is huge so that pooling data from all trials becomes infeasible. We propose a simpler two-stage procedure: {\tt CMMA}.
%%%%%%%%%%%%%%%%%%%%%%%%%%%%%%%%%%%%%%%%%%%%%%%%%%%%%%%%%%%%%%%%%%%%%%%%%%
\begin{algorithm}[H]
\caption{{\tt  CMMA}\label{algorithm: CMMA}}
\begin{algorithmic}[1]
\Input{$Y_i$, $T_i$, and $M_i$, $\vect{S_i}$, $H_k$}
\Output $\widehat{\beta}$

\For{trial $k=1$ to $K$}
\State Estimate {\tt ATE} of trial $k$ treatment on $Y$ using data from trial $k$ and denote it as ${\tt ATE}^{Y}_k$ .
\State Estimate {\tt ATE} of trial $k$ treatment on $M$ using data from trial $k$ and denote it as ${\tt ATE}^{M}_k$. 
\EndFor
\State Regress ${\tt ATE}^{Y}_k$ on ${\tt ATE}^{M}_k$ and additional trial-level covariates $H_k$.
\State Save the coefficient for ${\tt ATE}^{M}_k$ as $\widehat{\beta}$.
\State \Return{$\widehat{\beta}$}
\end{algorithmic}
\end{algorithm}
%%%%%%%%%%%%%%%%%%%%%%%%%%%%%%%%%%%%%%%%%%%%%%%%%%%%%%%%%%%%%%%%%%%%%%%%%%
% Firstly, for each trial, we estimate {\tt ATE}s on $Y$ and $M^k$. Most A/B tests compute {\tt ATE} on $Y$. Note that, here we need {\tt ATE} on $M$ (and {\tt ATE} on the higher order terms of $M$\footnote{Note: When underlying $\mu(m)$ has higher order terms, we need {\tt ATE} on the $k$th order term of mediator, $(\Delta M^k)_s$, rather than the $k$th order term of {\tt ATE} on mediator, ${(\Delta M)_s}^k$.}).  
% Secondly, we regress {\tt ATE} on $Y$ on {\tt ATE} on $M$ (and {\tt ATE} on the higher order terms of $M$) and additional trial-level covariates. The \textit{ordinary least square} (OLS) procedure with heteroscedasticity-consistent standard errors suffice. Wald test, a model selection tool, can be employed to decide the highest $k$th order term to include in the regression. The standard procedure is to run a regression with higher order terms and then test whether coefficients of those terms are zeros. See \citet[Chapter 5]{Greene2011EconometricAnalysis} for more technical discussions on Wald test.
Step~1-4 of {\tt CMMA} calculates {\tt ATE}s for each trial. {\tt ATE} on $Y$ and $M$ could be done in the standard procedure of online A/B tests.  Step~5 and~6 uses only trial-level summary statistics and covariates, making it very easy to implement. This  estimator has the same identification strategy as in Theorem~\ref{thm:identification} and is equivalent to a weight-adjusted 2SLS estimator. The proof is in Appendix~\ref{appendix:proof-iv-equivalent}.

Note that, {\tt CMMA} allows the same unit to participate in multiple trials.  We can always use regression/ANOVA with treatment interaction terms to estimate {\tt ATE}s of each trial for units in multiple trials, and then implement Step~5 and~6 of {\tt CMMA} on estimated {\tt ATE}s to get $\widehat{\beta}$. 

The most challenging part of applying {\tt CMMA} is finding valid trial-level characteristics $H_k$ to satisfy Assumption~\ref{assump:trial_independent}. A good $H_k$ should have explanatory power for treatment effects on outcome and mediator across trials.  However, similar to finding a valid instrument, there is no systematic way to produce $H_k$. Practitioners have to rely on available data and domain knowledge to argue for the validity of $H_k$. In Section~\ref{sec:simulation} Table~\ref{tab:assumption_violation}, we simulate the consequences of violating Assumption~\ref{assump:trial_independent}. In Section~\ref{sec:data} Figure~\ref{fig:ATE_NDCG}, we discuss the choice of $H_k$ in our real data application.  Future work is required on the sensitivity of {\tt CMMA} to $\matr{H}$.

\subsection{Hypothesis Testing}
In general, the reported standard errors from the second stage regression is slightly different from theoretical values without access to residuals in the first stage. But this becomes less of an issue as sample size increases.  Since the sample size is usually enormous in online A/B tests, we recommend using the reported standard errors in the second stage regression for convenience. 

Although we have assumed that $\mu_i(m) = \beta_i m$ for discussion to this point, the same proof is still valid for polynomial {\tt DRF}, $\mu_i(m) = \sum_{p=1}^{P}\beta_{i,p} m^p$. Let $\overline{\vect{M}_i} = [M_i,\cdots, M_i^{p},\cdots,M_i^{P}]^{\top}$ and $\overline{\vect{\beta}} = [\beta_1,\cdots,\beta_{p},\cdots,\beta_{P}]^{\top}$, where $\beta_{p} = \mathbb{E}[\beta_{i,p}]$. Then we can use $\overline{\vect{M}_i}^{\top}\overline{\vect{\beta}}$ in place of $M_i\beta$ in the proof. For our algorithm, in addition to ${\tt ATE}^{Y}_k$ and ${\tt ATE}^{M}_k$, we also need to estimate {\tt ATE} on the higher order terms of $M_i$, ${\tt ATE}^{(M^{p})}_k, p = 2,3,\cdots,P$. \footnote{Note: We need ${\tt ATE}^{(M^{p})}_k$, not the $p$-th order of ${\tt ATE}^{M}_k$: $({\tt ATE}^{M}_k)^{p}$. ${\tt ATE}^{(M^{p})}_k \neq ({\tt ATE}^{M}_k)^{p}$.}

To decide the highest $p$-th order term to include in the model of the second stage, we can use the common model selection tool: Wald test. The standard way is to run a regression with higher-order terms and then perform a series of tests to check whether coefficients of those terms are zeros. See \citet[Chapter 5]{Greene2011EconometricAnalysis} for more technical discussions on Wald test.

\section{Simulation}\label{sec:simulation}
We conduct Monte Carlo simulations to study the finite sample performance of our estimator. The details of our simulation set-up are described in Appendix~\ref{appendix:simulation}. The R code is available in our GitHub repository: \url{https://github.com/znwang25/cmma}.

\begin{table}[!htb]
    \caption{Finite-Sample Performance Comparison}\label{tab:estimator_performance}
        \resizebox{0.47\textwidth}{!}{%
    \begin{tabular}{lcccccc}
        \toprule
        \multicolumn{1}{c}{ } & \multicolumn{2}{c}{$N_{per}=200$} & \multicolumn{2}{c}{$N_{per}=500$} & \multicolumn{2}{c}{$N_{per}=1000$} \\
        \cmidrule(l{3pt}r{3pt}){2-3} \cmidrule(l{3pt}r{3pt}){4-5} \cmidrule(l{3pt}r{3pt}){6-7}
        Estimator & Bias & \makecell{ 95\% CI \\coverage } & Bias & \makecell{ 95\% CI \\coverage } & Bias & \makecell{ 95\% CI \\coverage }\\
        \midrule
        LIML & 0.003 & 93\% & -0.002 & 94\% & 0.000 & 94\%\\
        {\tt CMMA} & 0.058 & 67\% & 0.020 & 88\% & 0.011 & 88\%\\
        Full Sample 2SLS & 0.058 & 66\% & 0.020 & 83\% & 0.011 & 86\%\\
        \citet{Sobel2008IdentificationVariables} & 0.308 & 3\% & 0.279 & 3\% & 0.283 & 3\%\\
        LSEM \citep{Baron1986} & 0.937 & 0\% & 0.936 & 0\% & 0.937 & 0\%\\
        \bottomrule
        \end{tabular}
        }
\end{table}

We use Limited Information Maximum Likelihood (LIML) estimator as a benchmark to evaluate our {\tt CMMA} estimator.  An LIML estimator that is specified according to the simulation setup should have the best performance theoretically. We include two common estimators in the literature of mediation analysis into the comparison: \citet{Sobel2008IdentificationVariables} and LSEM~\citep{Baron1986}, which are derived under identification approaches discussed in Section~\ref{sec:literature_review}. \citet{Sobel2008IdentificationVariables} assumes complete mediation, and LSEM~\citep{Baron1986} relies on SI assumption~\citep{Imai2010}. Both assumptions are false under our setting. We also implement the Full Sample 2SLS estimator prescribed in Theorem~\ref{thm:identification}, which should produce similar results to {\tt CMMA}.

\begin{table}[h]
    \caption{Assumption Violation}\label{tab:assumption_violation}
    \resizebox{0.47\textwidth}{!}{%
    \begin{tabular}{lcccccc}
    \toprule
    \multicolumn{1}{c}{ } & \multicolumn{2}{c}{$N_{per}=200$} & \multicolumn{2}{c}{$N_{per}=500$} & \multicolumn{2}{c}{$N_{per}=1000$} \\
    \cmidrule(l{3pt}r{3pt}){2-3} \cmidrule(l{3pt}r{3pt}){4-5} \cmidrule(l{3pt}r{3pt}){6-7}
    Violation & Bias & \makecell{95\% CI \\coverage} & Bias & \makecell{95\% CI \\coverage} & Bias & \makecell{95\% CI \\coverage}\\
    \midrule
    None & 0.058 & 67\% & 0.020 & 88\% & 0.011 & 88\%\\
    A\ref{assump:relaxed-SI} & 0.065 & 65\% & 0.024 & 87\% & 0.009 & 93\%\\
    A\ref{assump:trial_independent} & 0.483 & 0\% & 0.470 & 0\% & 0.460 & 0\%\\
    A\ref{assump:first-stage} & 0.949 & 0\% & 0.931 & 0\% & 0.942 & 0\%\\
    \bottomrule
    \end{tabular}%
    }
\end{table}

\begin{table*}[!htb]
    \caption{Model Selection and Wald Tests}\label{tab:model_selection}\centering
        \resizebox{0.8\textwidth}{!}{%
    \begin{threeparttable}
        \begin{tabular}{lrrrrl}
            \toprule
            \multicolumn{1}{c}{ } & \multicolumn{3}{c}{\makecell{\% of Wald Tests Rejecting $H_0$ }} & \multicolumn{2}{c}{{\tt CMMA} Estimation} \\
            \cmidrule(l{3pt}r{3pt}){2-4} \cmidrule(l{3pt}r{3pt}){5-6}
            \makecell{$\mu(m)$: $\beta_1m+\beta_2 m^2 + \beta_3 m^3$ } & $H_0: \beta_3 = 0$ & $H_0: \beta_2 = \beta_3 = 0$ & $H_0: \beta_1=\beta_2=\beta_3= 0$ & \makecell{Highest Order of ${\tt ATE}^{(M^p)}_k$} & Estimated $\beta$'s of $\mu(m)$\\
            \midrule
            $4m$ & 3\% & 5\% & 100\% & 1 & 4.016\\
            $4m + 2m^2$ & 3\% & 100\% & 100\% & 2 & 4.019, 1.999\\
            $4m + 0m^2 + 5m^3$ & 100\% & 100\% & 100\% & 3 & 4.022, -0.001, 5\\
            \bottomrule
            \end{tabular}
            \begin{tablenotes}
                \item Note: The sample size per trial is 1000 and the number of trials is 100. $H_0$ is rejected if p-value $< 0.05$.
            \end{tablenotes}    
    \end{threeparttable}
    }
    \end{table*}
We set the sample size per trial ($N_{per}$), to 200, 500, and 1000 and perform 100 simulations for each setting. Table~\ref{tab:estimator_performance} reports average biases and  95\% confidence interval coverage of the estimators. The performance of {\tt CMMA} is quite good and largely comparable to LIML's result. When the sample size per trial is small, our estimator is slightly biased but the bias is much smaller than those of Sobel and LSEM estimators. As $N_{per}$ increases, the bias of our estimator decreases to zero, whereas the biases of the Sobel and LSEM estimators remain roughly the same. As the sample size is usually enormous in A/B tests (on average, one trial has millions of observations in the real data we obtained from an internet company), the bias of {\tt CMMA} will be negligible in practice. Table~\ref{tab:estimator_performance} also shows that, as $N_{per}$ increases, the 95\% confidence interval coverage of our estimator converges to the nominal coverage. This means that the OLS variance estimated in the trial-level regression is valid for hypothesis testing when $N_{per}$ is sufficiently large. In addition, the point estimate of Full Sample 2SLS estimator is numerically equal to {\tt CMMA}, which empirically validates the equivalence claim made in Section~\ref{sec:algorithm}.

In Table~\ref{tab:assumption_violation}, we examine the performance of {\tt CMMA} when each assumption fails. The results show that, the failure of Assumption~\ref{assump:relaxed-SI} does not seem to affect the estimator's unbiasedness, whereas with failures of Assumption~\ref{assump:trial_independent} or~\ref{assump:first-stage}, {\tt CMMA} is no longer unbiased. Comparing results across rows, it seems to suggest that violating Assumption~\ref{assump:first-stage} has worse consequence in terms of bias. Fortunately, Assumption~\ref{assump:first-stage} is testable as discussed in Section~\ref{sec:identification}.

We also test the performance of the Wald test with a higher order of $\mu(m)$. Table~\ref{tab:model_selection} shows that Wald tests can successfully select the correct model. For example, in the second row of Table~\ref{tab:model_selection}, the true mediator {\tt DRF} is a quadratic function of $m$: $4m + 2m^2$.  Wald tests in all simulations reject the null hypothesis that coefficients for the second-degree term and the third-degree term are zeros, whereas 97\% of simulations fail to reject the hypothesis that coefficient for the third-degree term is zero.  The result suggests that the highest order of ${\tt ATE}^{(M^p)}_k$ should be 2. The last column of Table~\ref{tab:wald_res} shows that, with a correctly specified model, we can accurately estimate all the underlying parameters. 

\section{Application}\label{sec:application}
We apply the approach on three most popular rank-aware evaluation metrics: {\tt NDCG}, {\tt MRR}, and {\tt MAP}, to show, for ranking algorithms that power search page of Etsy.com, which one could lead to the most significant lift of sitewide {\tt GMV}.  Since the offline A/B test literature~\citep{Gilotte2018OfflineSystems} bridges the inconsistency between changes of offline and online evaluation metrics, we only focus on, how sitewide {\tt GMV} would change for 10\% lifts in online {\tt NDCG}, {\tt MRR}, and {\tt MAP} of search page respectively.  All metrics in the application, unless otherwise noted, are \textit{online} metrics.  Please note the approach has not been deployed in Etsy.
% it is theoretical discussion, and real applications are likely to have many factors beyond the scope of this paper. 
This work is not intended to apply to, nor is it a prediction of, actual live performance metrics or performance changes on Etsy or any other property.
\subsection{User-Level Evaluation Metrics}\label{section: user-level-evaluation-metrics}
We follow the offline A/B test literature~\citep{Gilotte2018OfflineSystems} and define the three online rank-aware evaluation metrics at the user level.  Although the three metrics are originally defined at the query level in the test collection evaluation of information retrieval (IR) literature, the search page in the industry is an online product for users and thus the computation could be adapted to the user level.  More specifically, the three metrics are constructed as follows: 1) query-level metrics are computed using rank positions on search page and user conversion status as binary relevance, and non-conversion associated queries have zero values\footnote{If the user purchases the item that she has clicked on the search page, the relevance is 1; otherwise 0.}; 2) user-level metrics is the average of query-level metrics across all queries the user issues (including non-conversion associated queries), and users who do not search or convert have zero values.  Also, all the three metrics are defined at rank position 48, the lowest position of the first page of search results in Etsy.com.
\subsection{Data}\label{sec:data}
We have access to summarized results of 190 randomly selected experiments from the online A/B test platform of Etsy.com.  All the experiments in the data have the user as an experimental unit.  The data include descriptive information about each experiment such as the tested product change, the product team that initiated the experiment, and summary statistics of each experiment such as average (user-level) {\tt NDCG} per user in treatment and control groups.  Note that, the difference between the average metric per user in treatment and control groups is {\tt ATE} on the metric (Step~1~-~4 in Algorithm~\ref{algorithm: CMMA}).
% It does not include {\tt ATE} on {\tt NDCG}'s higher-order terms, however. So we use the raw data to calculate {\tt ATE} on {\tt NDCG$^2$} and {\tt NDCG$^3$} for each experiment.

We use the taxonomy of product teams as our trial-level covariates $\matr{H}$ because team taxonomy is quite informative of experiments as Figure~\ref{fig:ATE_NDCG} suggests.  Figure~\ref{fig:ATE_NDCG} shows the density of {\tt ATE} on {\tt NDCG} for selected teams including UI/UX design, marketing, search ranking. First, it is evident that distributions of {\tt ATE} on {\tt NDCG} vary by team, which implies that Assumption~\ref{assump:trial_independent} is likely to be satisfied.  In particular, most experiments of the search ranking team post positive gains on {\tt NDCG}, whereas experiments of the UI/UX team barely affects {\tt NDCG}.
% Other experiments by Payment or Marketing teams, despite not directly changing ranking algorithms, also influence {\tt NDCG} to some degree.
Second, within each team there are significant variances of {\tt ATE} on {\tt NDCG}, which suggests that Assumption~\ref{assump:first-stage} holds. Due to page limitation, only results for {\tt NDCG} are presented here, but the distributions of {\tt ATE} on {\tt MAP} and {\tt MRR} exhibit the similar pattern.
\begin{figure}[h]
    \includegraphics[width = 0.3\textwidth]{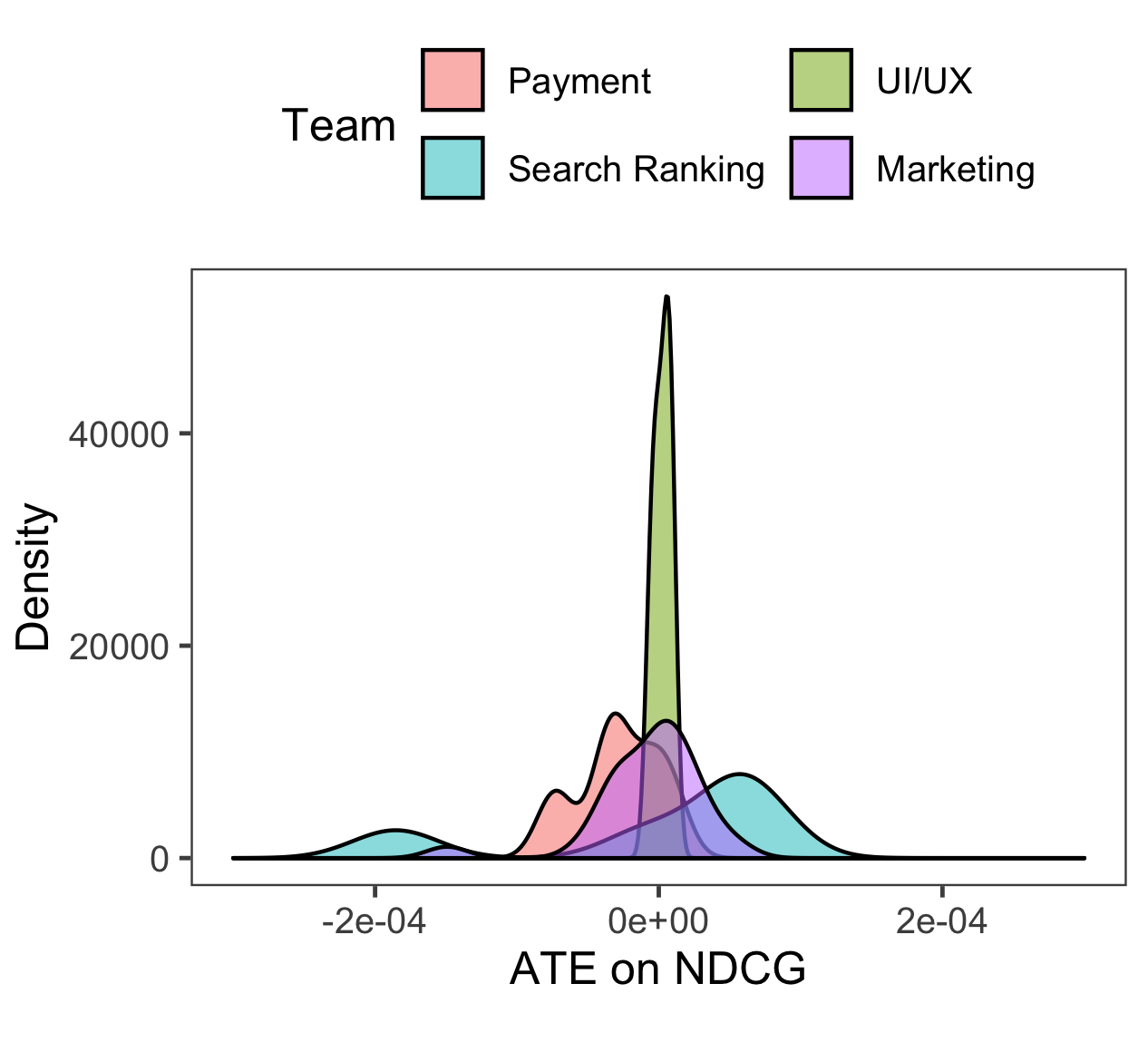}
    \caption{Density of {\tt ATE} on {\tt NDCG}}
    \label{fig:ATE_NDCG}
\end{figure}
\subsection{Results}
\begin{table}[h]
    \centering
    \caption{Wald Test Results}\label{tab:wald_res}
    \begin{tabular}{lccc}
        \toprule
        \multicolumn{1}{c}{ } & \multicolumn{3}{c}{P-Value} \\
        \cmidrule(l{3pt}r{3pt}){2-4}
        Null Hypothesis &$M$: NDCG & $M$: MAP & $M$: MRR\\
        \midrule
        $H_0: \beta_3 = 0$ & 0.00119 & 0.02480 & 0.00346\\
        $H_0: \beta_2 = \beta_3 = 0$ & 0.00018 & 0.00003 & 0.00000\\
        $H_0: \beta_1=\beta_2=\beta_3= 0$ & 0.00018 & 0.00003 & 0.00000\\
        \bottomrule
        \end{tabular}
\end{table}
To decide the polynomial terms in the model, we perform Wald tests. The result from Wald tests is in Table~\ref{tab:wald_res}. Since all the null hypothesis are rejected, the results suggest us including {\tt ATE} on $M$ , {\tt ATE} on $M^2$ and {\tt ATE} on $M^3$ in the model of each $M$ ({\tt NDCG}, {\tt MAP}, {\tt MRR}).
\begin{figure*}[h]
    \includegraphics[width = 0.65\textwidth]{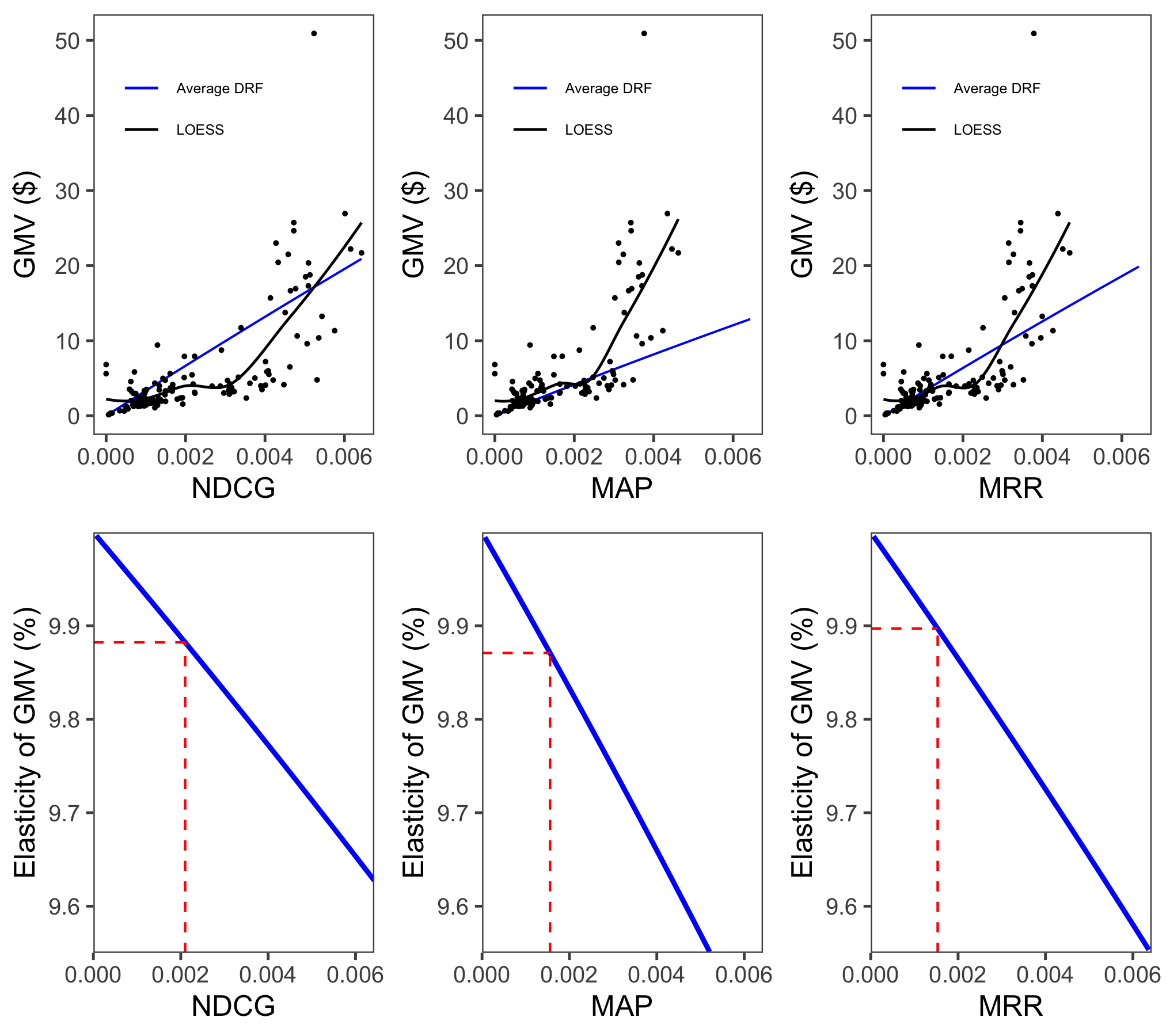}
    \caption{Estimated Mediator Dose Response Function}
    \label{fig:GMS_res_6plot}
\end{figure*}

Figure~\ref{fig:GMS_res_6plot} shows the estimation results.  On the first row, blue lines depict estimated mediator {\tt DRF}s.  Scattered points represent summary statistics (the data for {\tt CMMA}), {\tt ATE} on ranking evaluation metrics and {\tt ATE} on {\tt GMV}, of all experiments.  Black curves show results from fitting the data by LOESS.  Note that, the range of three evaluation metrics could be much smaller than those in IR literature since they are defined at the user level (see Section~\ref{section: user-level-evaluation-metrics}).  The estimated coefficients of mediator {\tt DRF} are in Table~\ref{tab: appendix_reg_res} in the Appendix.  The estimated mediator {\tt DRF}s show that all three online metrics have positive causal effects on {\tt GMV}.  Note that, the causal relationships are different from the pattern of the data (scatter points and black curves).  The differences between blue lines and black curves show the bias from fitting the data by machine learning methods without addressing omitted variables.

The second row of Figure~\ref{fig:GMS_res_6plot} shows the \textit{elasticity} of {\tt GMV}: the percentage change of {\tt GMV} for a 10\% increase in each evaluation metric at its different values, which are derived from estimated mediator {\tt DRF}.  The downward slopes imply that, for all three evaluation metrics, as they increase, the benefit of continuous improving them on {\tt GMV} decreases. For example, when average {\tt NDCG} per user equals 0.0021, its 10\% increase leads to a 9.88\% increase in average {\tt GMV} per user. Yet, when it equals 0.006, its 10\% increase only leads to a 9.65\% increase in average {\tt GMV} per user.

Now it is easy for product owners to pick the evaluation metric that could guide algorithm development to achieve the most significant lift in online {\tt GMV}.  Suppose the current average values (per user) of {\tt NDCG}, {\tt MAP}, and {\tt MRR} from live data are 0.00210, 0.00156, and 0.00153 respectively.  From estimated mediator {\tt DRF}s, we can calculate their corresponding elasticities of {\tt GMV}: 9.88\%, 9.87\%, and 9.90\%, which are marked by red lines in Figure~\ref{fig:GMS_res_6plot}.  Because online {\tt MRR} has the highest elasticity of {\tt GMV}, we should choose offline {\tt MRR}, which is estimated based on offline A/B test literature~\citep{Gilotte2018OfflineSystems} and thus has the same move as online {\tt MRR}, to guide the development of ranking algorithms.

% The red lines mark the current average level of {\tt NDCG},{\tt MAP} and {\tt MRR} and their corresponding elasticities in our sample. At the current level, a 10\% increase in {\tt MRR} leads to the highest increase in {\tt GMV}, a 9.9\% increase. 

%   elasticity_avg search_purchase_NDCG48_mean
%            <dbl>                       <dbl>
% 1          0.988                     0.00210
%   elasticity_avg search_purchase_AP48_mean
%            <dbl>                     <dbl>
% 1          0.987                   0.00156
%   elasticity_avg search_purchase_RR48_mean
%            <dbl>                     <dbl>
% 1          0.990                   0.00153
% Our results assure product managers and engineers that training ranking algorithms for high {\tt NDCG} is beneficial to the business.  We offline quantifies the {\tt GMV} impact of {\tt NDCG} so that they can cross out the development of ranking algorithms that has low {\tt GMV} boost without costly A/B tests.  

\section{Conclusion}
In the internet industry, the algorithms developed offline power online products and online products contribute to the success of a business.  In many cases, offline evaluation metrics, which guide algorithm development, are different from online business KPIs.  It is important for product owners to pick the offline evaluation metric guided by which the algorithm could maximize online business KPIs.  By noticing that online products could be assessed by online counterparts of offline evaluation metrics, we decompose the problem into two parts.  Since the offline A/B test literature works out the first part: counterfactual estimators of offline evaluation metrics that move the same way as their online counterparts, we focus on the second part: inferring causal effects of online evaluation metrics on business KPIs.  The offline evaluation metric whose online counterpart causes the most significant lift in online business KPIs should be the north star.  We model online evaluation metrics as mediators and formalize the problem as to identify, estimate, and test mediator {\tt DRF}.  Our novel approach {\tt CMMA} combines mediation analysis and meta-analysis and has many advantages over the two strands of literature.  In particular, it takes as inputs only summary statistics from multiple past A/B tests, and thus it is easy to implement in scale.  We apply the approach on Etsy's real data to uncover the causality between three most popular rank-aware online evaluation metrics and {\tt GMV}, and show how we successfully identify {\tt MRR} as the offline evaluation metric for {\tt GMV} maximization.
% \begin{equation}
%   \lim_{n\rightarrow \infty}x=0
% \end{equation}

% \begin{displaymath}
%   \sum_{i=0}^{\infty} x + 1
% \end{displaymath}

% \begin{figure}[h]
%   \centering
%   \includegraphics[width=\linewidth]{sample-franklin}
%   \caption{1907 Franklin Model D roadster. Photograph by Harris \& Ewing, Inc. [Public domain], via Wikimedia Commons. (\url{https://goo.gl/VLCRBB}).}
%   \Description{The 1907 Franklin Model D roadster.}
% \end{figure}
%
% The acknowledgments section is defined using the "acks" environment (and NOT an unnumbered section). This ensures
% the proper identification of the section in the article metadata, and the consistent spelling of the heading.
% \begin{acks}
% \end{acks}
% The next two lines define the bibliography style to be used, and the bibliography file.
\bibliographystyle{ACM-Reference-Format}
\balance 
\bibliography{references}
\clearpage
% If your work has an appendix, this is the place to put it.
\appendix
\setcounter{table}{0}
\renewcommand{\thetable}{\Alph{section}\arabic{table}}

\section{Proof of Theorem~\ref{thm:identification}}\label{appendix:proof-thm}
\begin{proof}
    Let $\matr{X_i}= \begin{bmatrix}
        (\vect{S_i}T_{i})^{\top}\vect{\tau}  & \vect{S_i}^{\top} & (\vect{S_i}T_{i})^{\top}\vect{H}
    \end{bmatrix}$, and \\ $\vect{B}= \begin{bmatrix}
        \beta  & \vect{\theta} +  \vect{\phi}\beta & \vect{\pi}
    \end{bmatrix}^{\top}$, so Equation \ref{eq: iv_obs_y1} can be written as  
$$ Y_i = \matr{X_i}\vect{B} + \epsilon_i'$$

Because of assumption \ref{assump:first-stage}, $\mathbb{E}[\matr{X_i}^{\top}\matr{X_i}]$ is guaranteed to have full rank, thus is invertible. If we are able to show that $\mathbb{E}[\matr{X_i}^{\top}\epsilon_i'] = \vect{0}$, then $\vect{B}$ can be estimated by linear projection $\vect{B}^{*} = \mathbb{E}[\matr{X_i}^{\top}\matr{X_i}]^{-1}\mathbb{E}[\matr{X_i}^{\top}Y_i]$. While this is infeasible, we can first estimate $\vect{\tau}$ with its estimate $\widehat{\vect{\tau}}$ and use $\matr{\widehat{X_i}}= \begin{bmatrix}
    (\vect{S_i}T_{i})^{\top}\widehat{\vect{\tau}}  & \vect{S_i}^{\top} & (\vect{S_i}T_{i})^{\top}\vect{H}
\end{bmatrix}$ in place of $\vect{X_i}$. Let $\mathbf{\widehat{X}}$ be N-component data vector with ith element $\matr{\widehat{X_i}}$. The resulted estimator is the two-stage least squares(2SLS) estimator, 
$$\widehat{\vect{B}}_{2SLS} = (\mathbf{\widehat{X}}^{\top}\mathbf{\widehat{X}})^{-1}\mathbf{\widehat{X}}^{\top}\mathbf{Y}$$ to estimate $\vect{B}$.

Note that, because of Assumption \ref{assump:random-ts}, \ref{assump:relaxed-SI} and \ref{assump:trial_independent}, the following equations are true.
\begin{align*}
    &\mathbb{E}[\eta_i] = \mathbb{E}[\vect{S_i}T_{i}]^{\top}\mathbb{E}[(\vect{\tau_i}-\vect{\tau)}] + \mathbb{E}[M_{i}^{*}] = 0\\
    &\mathbb{E}[M_{i}(\beta_i-\beta)] =  \mathbb{E}[\mathbb{E}[M_{i}(\beta_i-\beta)|\vect{S}_i,T_i]]= \mathbb{E}[M_{i}\mathbb{E}[(\beta_i-\beta)|\vect{S}_i,T_i]] = 0\\
    &\mathbb{E}[\vect{\gamma}_i-\vect{H}\vect{\pi}]=\mathbb{E}[\mathbb{E}[\vect{\gamma}_i|\vect{H}]-\vect{H}\vect{\pi}]= \vect{0} 
\end{align*}
With the first component in $\vect{X_i}$, 
    \begin{align*}
        &\mathbb{E}\left[(\vect{S_i}T_{i})^{\top}\vect{\tau}\epsilon_i'\right]\\
        =& \mathbb{E}\left[(\vect{S_i}T_{i})^{\top}\vect{\tau}\left[(\vect{S}_iT_{i})^{\top}(\vect{\gamma}_i-\vect{H}\vect{\pi})+ \eta_i\beta + M_{i}(\beta_i-\beta)+Y_i^{*}\right]\right]\\
        =& \mathbb{E}\left[(\vect{S_i}T_{i})^{\top}\vect{\tau}(\vect{S}_iT_{i})^{\top}(\vect{\gamma}_i-\vect{H}\vect{\pi})\right]\\
        &\;+\mathbb{E}\left[(\vect{S_i}T_{i})^{\top}\vect{\tau}\right]\mathbb{E}\left[\eta_i\beta + M_{i}(\beta_i-\beta)+Y_i^{*}\right]\\
        =&\mathbb{E}\left[(\vect{S_i}T_{i})^{\top}\vect{\tau}(\vect{S}_iT_{i})^{\top}(\vect{\gamma}_i-\vect{H}\vect{\pi})\right] \\
        =&\mathbb{E}\left[\sum_{k=1}^{K}\mathbb{E}\left[S_{i,k}T_{i}\tau_{i,k}(\gamma_{i,k}-H_{k}\vect{\pi})\right]\right]\\
        =&\mathbb{E}\left[\sum_{k=1}^{K}\mathbb{E}\left[\mathbb{E}\left[S_{i,k}T_{i}\tau_{i,k}(\gamma_{i,k}-H_{k}\vect{\pi})|H_{k}\right]\right]\right]\\
        =&\mathbb{E}\left[\sum_{k=1}^{K}\mathbb{E}\left[S_{i,k}T_{i}\tau_{i,k}\left(\mathbb{E}\left[\gamma_{i,k}|H_{k}\right]-H_{k}\vect{\pi}\right)\right]\right]\\
        =& 0,
    \end{align*}
where the forth equality expands the matrix multiplication and the sixth equality follows from Assumption \ref{assump:trial_independent}.

For the second component,    
\begin{align*}
    &\mathbb{E}\left[\vect{S_i}\epsilon_i'\right]\\
    =& \mathbb{E}\left[\vect{S_i}(\vect{S}_iT_{i})^{\top}(\vect{\gamma}_i-\vect{H}\vect{\pi})\right] \\
    & \;+ \mathbb{E}\left[\vect{S_i}\right]\mathbb{E}\left[\eta_i\beta + M_{i}(\beta_i-\beta)+Y_i^{*}\right]\\
    =& \mathbb{E}\left[\vect{S_i}(\vect{S}_iT_{i})^{\top}\right]\mathbb{E}\left[\vect{\gamma}_i-\vect{H}\vect{\pi}\right]\\
    =& 0,
\end{align*}
where both the second and third equality follows from Assumption \ref{assump:random-ts}.

With respect to the third component,
\begin{align*}
    &\mathbb{E}\left[(\vect{S_i}T_{i})^{\top}\vect{H}\epsilon_i'\right]\\
    =& \mathbb{E}\left[(\vect{S_i}T_{i})^{\top}\vect{H}(\vect{S}_iT_{i})^{\top}(\vect{\gamma_i}-\vect{H}\vect{\pi})\right] \\
    & \; +\mathbb{E}\left[(\vect{S_i}T_{i})^{\top}\vect{H}\right]\mathbb{E}\left[\eta_i\beta + M_{i}(\beta_i-\beta)+Y_i^{*}\right]\\
    =& \mathbb{E}\left[(\vect{S_i}T_{i})^{\top}\vect{H}(\vect{S}_iT_{i})^{\top}\mathbb{E}\left[(\vect{\gamma_i}-\vect{H}\vect{\pi})|\vect{H}\right]\right]\\
    = & 0,
\end{align*}
where both the second and third equality follow from Assumption \ref{assump:random-ts}.

Taken together, we have shown that $\mathbb{E}[\matr{X_i}^{\top}\epsilon_i'] = \vect{0}$, therefore $\vect{B}$ can be estimated by $$\widehat{\vect{B}}_{2SLS} = (\mathbf{\widehat{X}}^{\top}\mathbf{\widehat{X}})^{-1}\mathbf{\widehat{X}}^{\top}\mathbf{Y}.$$
\end{proof}

\section{Equivalence of {\tt CMMA} and weight-adjusted 2SLS estimator}\label{appendix:proof-iv-equivalent}

Let $\mathbf{Z}$ be an N-component data vector with ith element $(\vect{S_i}T_i)^{\top}$, $\mathbf{S}$ be an N-component data vector with ith element $\vect{S_i}^{\top}$, and $\mathbf{P_{S}} = \mathbf{S}(\mathbf{S}^{\top}\mathbf{S})^{-1}\mathbf{S}^{\top}$, $\mathbf{M_{S}} = \mathbf{I} - \mathbf{P_{S}}$. The following proposition shows the link between {\tt CMMA} and 2SLS estimator. 
\begin{proposition}\label{proposition:iv-equivalent}
    The {\tt CMMA} method defined in Algorithm~\ref{algorithm: CMMA} is equivalent to a weight-adjusted 2SLS IV estimator with weight $\mathbf{W}^{\top}\mathbf{W}$, where $\mathbf{W}=(\mathbf{Z}^{\top}\mathbf{M_{S}}\mathbf{Z})^{-1}\mathbf{Z}^{\top}$
\end{proposition}

\begin{proof}
    % Let $\mathbf{Z}$ be N-component data vector with ith element $(\vect{S_i}T_i)^{\top}$, $\mathbf{S}$ be N-component data vector with ith element $\vect{S_i}^{\top}$. Let $\mathbf{P_{S}} = \mathbf{S}(\mathbf{S}^{\top}\mathbf{S})^{-1}\mathbf{S}^{\top}$, $\mathbf{M_{S}} = \mathbf{I} - \mathbf{P_{S}}$.
Let $\vect{\beta} = \begin{pmatrix} \beta \\ \vect{\pi} \end{pmatrix}$, $\mathbf{\overline{Z}} = \begin{bmatrix} \vect{\mathbf{Z}\widehat{\tau}} & \mathbf{Z}\vect{H} \end{bmatrix}$ . 
    Since we are not particularly interested in coefficients in front of $\vect{S_i}$, using Frisch–Waugh–Lovell theorem, we can get 2SLS estimator for $\vect{\beta}$,
    $$ \widehat{\vect{\beta}}_{2SLS} = \left(
         \mathbf{\overline{Z}}^{\top}\mathbf{M_{S}}\mathbf{M_{S}}        \mathbf{\overline{Z}}\right)^{-1} \mathbf{\overline{Z}}^{\top}\mathbf{M_{S}}\mathbf{M_{S}}\mathbf{Y}.
        $$ 
    We could use a weighting matrix $\mathbf{W}$ and still get a consistent estimator for $\vect{\beta}$,
    $$\tilde{\vect{\beta}} = (\mathbf{\overline{Z}}^{\top}\mathbf{M_{S}}\mathbf{W}^{\top}\mathbf{W}\mathbf{M_{S}}\mathbf{\overline{Z}})^{-1}\mathbf{\overline{Z}}^{\top}\mathbf{M_{S}}\mathbf{W}^{\top}\mathbf{W}\mathbf{M_{S}}\mathbf{Y}$$.

    Use $\mathbf{W} = (\mathbf{Z}^{\top}\mathbf{M_{S}}\mathbf{Z})^{-1}\mathbf{Z}^{\top}$

    \begin{align*}
        \mathbf{\overline{Z}}^{\top}\mathbf{M_{S}}\mathbf{W}^{\top} & = 
        \begin{pmatrix}
            \vect{\widehat{\tau}}^{\top}\mathbf{Z}^{\top}\mathbf{M_{S}}\mathbf{Z}(\mathbf{Z}^{\top}\mathbf{M_{S}}\mathbf{Z})^{-1} \\
            \vect{H}^{\top}\mathbf{Z}^{\top}\mathbf{M_{S}}\mathbf{Z}(\mathbf{Z}^{\top}\mathbf{M_{S}}\mathbf{Z})^{-1}
    \end{pmatrix}  
        = \begin{pmatrix}
            \vect{\widehat{\tau}}^{\top} \\ \vect{H}^{\top}
        \end{pmatrix} 
    \end{align*}
    $$\widehat{\vect{\beta}}_{{\tt CMMA}} = (\begin{pmatrix}
        \vect{\widehat{\tau}}^{\top} \\ \vect{H}^{\top}
        \end{pmatrix} 
     \begin{pmatrix}
            \vect{\widehat{\tau}} & \vect{H}
        \end{pmatrix})^{-1} \begin{pmatrix}
        \vect{\widehat{\tau}}^{\top} \\ \vect{H}^{\top}
    \end{pmatrix} \mathbf{W}\mathbf{M_{S}}\mathbf{Y} $$

$\mathbf{W}\mathbf{M_{S}}\mathbf{Y}= (\mathbf{Z}^{\top}\mathbf{M_{S}}\mathbf{Z})^{-1}\mathbf{Z}^{\top}\mathbf{M_{S}}\mathbf{Y}$ is equivalent to regressing $Y_i$ on $\vect{S_i}T_i$ and $S_i$ and take coefficients of $\vect{S_i}T_i$. And it is equivalent to regress on $T_i$ for each trial.
\end{proof}
\section{Simulation Setup}\label{appendix:simulation}
We follow the specification described in Equation \ref{eq: obs_m} and Equation \ref{eq: iv_obs_y1} and let the $M^{*}_i$ and $Y^{*}_i$ be jointly normally distributed:
\begin{align*}
\begin{pmatrix}
    M^{*}_i \\
    Y^{*}_i
\end{pmatrix}
\sim \mathcal{N}
\left(
\begin{pmatrix}
    {0} \\
    {0}
\end{pmatrix}
,
\begin{bmatrix}
    \sigma_{M^*}^2 &,& \rho\sigma_{M^*}\sigma_{Y^*}\\
    \rho\sigma_{M^*}\sigma_{Y^*} &,& \sigma_{Y^*}^2 
    \end{bmatrix}
\right)
\end{align*}
with $(\rho, \sigma_{M^*}, \sigma_{Y^*})=(0.95, 3, 3)$. We fix the number of trials to be 50, and used independent uniform distributions to specify parameter value for each element of $\vect{\theta}$, $\vect{\phi}$. To satisfy Assumption \ref{assump:trial_independent} and \ref{assump:first-stage}, $\vect{\tau}$ is set to be the sum of a $H$ dependent term and a random vector drawing from a uniform distribution. All the other parameter values are listed in the Table \ref{tab:param_value}. The innovations in the error terms such as $\beta_i-\beta$, elements of $\vect{\tau}_i-\vect{\tau}$, and 
$\vect{\gamma_{i}} - \vect{H}\vect{\pi}$ are all drawn from independent normal distributions $\mathcal{N}(0, 0.5)$.
  \begin{table}[H]
    \caption{Parameter Values}
    \label{tab:param_value}  
    \centering%
\begin{tabular}{c|c}
{\bf Parameters} & {\bf Value} \\
\midrule
$\beta$ & 4   \\
$\vect{\pi}$ & $(0, 1.5, 3)^{\top}$\\
$\vect{\theta}$ & $(\theta_1, \cdots, \theta_{50})^{\top}$,  $\theta_k \sim \mathcal{U}(-2, 2) \; \forall k$  \\
$\vect{\phi}$ & $(\phi_1, \cdots, \phi_{50})^{\top}$,  $\phi_k \sim \mathcal{U}(-2, 2) \; \forall k$  \\
$\vect{\tau}$ & \makecell{$\vect{\Tau}^{\top}\matr{H}^{\top}+(\tau_1, \cdots, \tau_{50})^{\top}$, \\ $\tau_k \sim \mathcal{U}(-3, 3) \;\forall k$ and $\vect{\Tau}$ = $(0.5, 1, 2.5)^{\top}$} \\
\bottomrule
\end{tabular}
\end{table}

The one-hot group assignment variable $\vect{S}_i$ and treatment indicator $T_i$ are randomly generated. We assume the 50 trials can be grouped into 3 experiment types and use experiment types as trial level covariates. Trials are randomly assigned into three types and $\matr{H}$ is a $3\times 50$ matrix representing such assignment. Under this setup, all the assumptions are satisfied, and thus we are ready to estimate $\beta$.

\section{Tables}\label{appendix:tables}
\begin{table}[!htbp] \centering 
    \caption{Second Stage Regression Results} 
    \label{tab: appendix_reg_res} 
  \begin{tabular}{@{\extracolsep{5pt}}lccc} 
    \toprule
   & \multicolumn{3}{c}{\textit{Dependent variable: GMS}} \\ 
  \cline{2-4} 
  \\[-1.8ex] & $M$: NDCG & $M$: MPP & $M$: MRR \\ 
  \\[-1.8ex] & (1) & (2) & (3)\\ 
  \hline \\[-1.8ex] 
  ${\tt ATE}^{M}$ & 3,369.9$^{*}$ & 2,113.3 &3,227.6  \\ 
  & (1,753.3) & (2,600.2) & (2,368.2) \\ 
 ${\tt ATE}^{M^2}$ & $-$18,593.8$^{***}$ & $-$17,254.4$^{*}$ &$-$21,411.1$^{**}$ \\ 
  & (6,606.5) & (9,718.6)  &(8,882.6)  \\ 
 ${\tt ATE}^{M^3}$ & 16,733.0$^{***}$ & 16,191.1$^{**}$ &19,229.7$^{***}$   \\ 
  & (5,162.6) & (7,213.5) &(6,576.5)  \\ 
  \hline \\[-1.8ex] 
  Observations & 190 & 169 & 190 \\ 
  R$^{2}$ & 0.240 & 0.271 & 0.282 \\ 
  \bottomrule
  \textit{Note:}  & \multicolumn{3}{r}{$^{*}$p$<$0.1; $^{**}$p$<$0.05; $^{***}$p$<$0.01} \\ 
  \end{tabular} 
  \end{table} 
\end{document}